\PassOptionsToPackage{dvipsnames}{xcolor}

\documentclass[prb,preprint,onecolumn]{revtex4-1} 

\usepackage{amsmath}
\usepackage{amsfonts}
\usepackage{graphicx}
\usepackage{esvect}
\usepackage{mathtools}
\usepackage{subcaption}
\usepackage{xcolor}
\usepackage{appendix}

\usepackage{siunitx} 
\usepackage{placeins}

\usepackage{float}

\usepackage{pbox}

\usepackage{url}

\usepackage{xcolor}
\pagecolor{white}

\def\Z{\mathbb{Z}}

\newcounter{thmcount}
\newenvironment{theorem}[1][]{
  \refstepcounter{thmcount}
  \par\medskip\noindent
  \textbf{Theorem \thethmcount.}
  \ifx\relax#1\relax\else\ (\textit{#1})\fi\ \itshape
}{
  \par\medskip
}

\newenvironment{proof}{
  \par\medskip\noindent
  \textit{Proof.}
  \enspace
}{
  \hfill\ensuremath{\square}
  \par\medskip
}

\newcounter{remcount}
\newenvironment{remark}[1][]{
  \refstepcounter{remcount}
  \par\medskip\noindent
  \textit{Remark \theremcount.}
  \if\relax\detokenize{#1}\relax\else\ (\textit{#1})\fi\
}{
  \par\medskip
}

\newcommand{\Norm}{\text{Norm}}

\newcommand{\Zi}{\mathbb{Z}[i]}

\newcommand{\lcm}{\text{lcm}}

\begin{document}

\title{Transitions with the same energy difference in the Bohr model of the hydrogen atom}

\author{Matthias Reinsch}
\email{mreinsch@sdsu.edu}
\affiliation{Department of Physics, San Diego State University, San Diego, CA 92182, USA}

\begin{abstract}

In the Bohr model of the hydrogen atom~\cite{bohr1913},
the energy levels are a negative constant divided by the square of the level number.
It is well known that special pairs of transitions exist that have the same energy difference, and a systematic treatment of this is given in the paper by 
Do and Phan~\cite{DoPhan2022}.

In this paper we describe a simple method (using equal norms of Gaussian integers, and the Brahmagupta--Fibonacci identity)
for finding all such pairs of transitions. We also analyze cascades of equal-frequency transitions, and use a theorem due to Fermat to show
that cascades with more than three levels are not possible.

We conclude the paper by going beyond Bohr's 1913 model, and analyzing some Diophantine equations related to the
nonrelativistic Schr\"odinger Equation for hydrogen, some of which are similar to the Diophantine equations studied in the main part of the paper.

\end{abstract}

\date{\today}

\maketitle

\section{Introduction}

The Bohr energy levels are the subject of the main part of this paper, not the energy levels of more complicated theories that include relativistic and other effects.
As is widely recognized, it is of interest to look at the relationship between mathematics and physics in many different contexts. The physics that we discuss in
the main part of this paper is that of Bohr's 1913 model.
The same energy levels are gotten from the nonrelativistic Schr\"odinger Equation, discussed in Section~\ref{sec_SE},
and in that section we analyze further Diophantine equations that involve the angular momentum quantum number.
Some of those Diophantine equations are similar to the Diophantine equations studied in the main part of the paper.

Do and Phan~\cite{DoPhan2022} showed how all equal-frequency transition pairs can be obtained from solutions of a
Diophantine equation, using a rational parametrization of an associated conic.
Another excellent reference is ``Question 06/00: Hydrogen Atom'' on Yuval Kantor's problem page~\cite{KantorQuestion2000}.
The corresponding answer page records contributions by several respondents,
including Kaplan and Liu, Ussishkin, Sundaram, and Perondi~\cite{KantorAnswer2000}.

\section{A Diophantus--Brahmagupta--Fibonacci Parametrization of Equal-Frequency Hydrogen Transitions}

It is customary to use the term ``Brahmagupta--Fibonacci identity'' for an identity first proved by Diophantus of Alexandria. A better term would be
``Diophantus--Brahmagupta--Fibonacci identity,'' but this is cumbersome, so we will use the shorter form of the name for the identity.

\begin{figure}[!htb]
\centering
\includegraphics[width=0.7\textwidth]{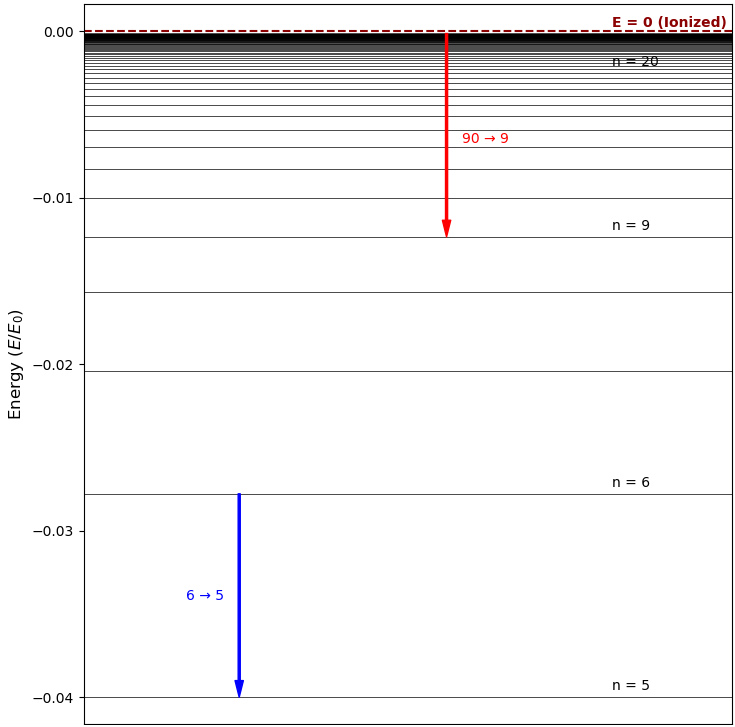}
\caption{Example: The two transitions shown produce the same photon frequency because 1/25 - 1/36 = 1/81 - 1/8100.}
\label{fig1}
\end{figure}

In the Bohr model of the hydrogen atom, the energy levels are proportional to
\[
   -\frac{1}{n^2}.
\]
(An inverse-square difference structure is present in Balmer's
empirical formula for the visible hydrogen spectrum~\cite{Balmer1885}, and
Rydberg~\cite{Rydberg1890} later generalized this structure to a broader spectral formula.
However, we cite Bohr's 1913 model as the central reference because it introduced the concept of
atomic energy levels, and we are interested in relationships between mathematics and physical theory.)
Two transitions
\[
   N_1\to n_1,
   \qquad
   N_2\to n_2,
   \qquad
   N_1>n_1,
   \quad
   N_2>n_2,
\]
produce photons with the same frequency precisely when
\begin{equation}\label{eq:freq}
   \frac{1}{n_1^2}-\frac{1}{N_1^2}
   =
   \frac{1}{n_2^2}-\frac{1}{N_2^2}.
\end{equation}
Equivalently,
\begin{equation}\label{eq:reciprocal-sum}
   \frac{1}{n_1^2}+\frac{1}{N_2^2}
   =
   \frac{1}{N_1^2}+\frac{1}{n_2^2}.
\end{equation}
Thus we are interested in
equalities between two sums of
reciprocal squares.
(Note that cases of equal-frequency transition pairs are mathematically equivalent to
degeneracies for the problem of two non-interacting hydrogen atoms in the Bohr model.)

One of the main objectives of this paper is to record a simple parametrization, based on the
Brahmagupta--Fibonacci identity, and to explain why it generates all
positive-integer equal-frequency transition pairs.

\subsection{The introduction of four parameters}

For integers $r,s,u,v$, the Brahmagupta--Fibonacci identity says
\begin{equation}\label{eq:bf}
   (r^2+s^2)(u^2+v^2)
   =
   (ru-sv)^2+(rv+su)^2
   =
   (ru+sv)^2+(su-rv)^2.
\end{equation}
Equivalently, in Gaussian-integer language~\cite{NivenZuckermanMontgomery1991, ConradGaussianIntegers},
\[
   (r+is)(u+iv)=(ru-sv)+i(rv+su),
\]
whereas
\[
   (r+is)(u-iv)=(ru+sv)+i(su-rv).
\]
Both Gaussian integers have the same norm, namely
\[
   (r^2+s^2)(u^2+v^2).
\]

Thus every choice of $r,s,u,v$ gives an identity of the form
\begin{equation}\label{eq:sum-two-squares}
   A^2+B^2=C^2+D^2,
\end{equation}
where
\begin{equation}\label{eq:ABCD}
   A=ru-sv,
   \qquad
   B=rv+su,
   \qquad
   C=ru+sv,
   \qquad
   D=su-rv.
\end{equation}
Signs and order are not essential; for positive-integer applications one may
take absolute values and interchange the two entries in either pair.

\begin{figure}[!htb]
\centering
\includegraphics[width=0.7\textwidth]{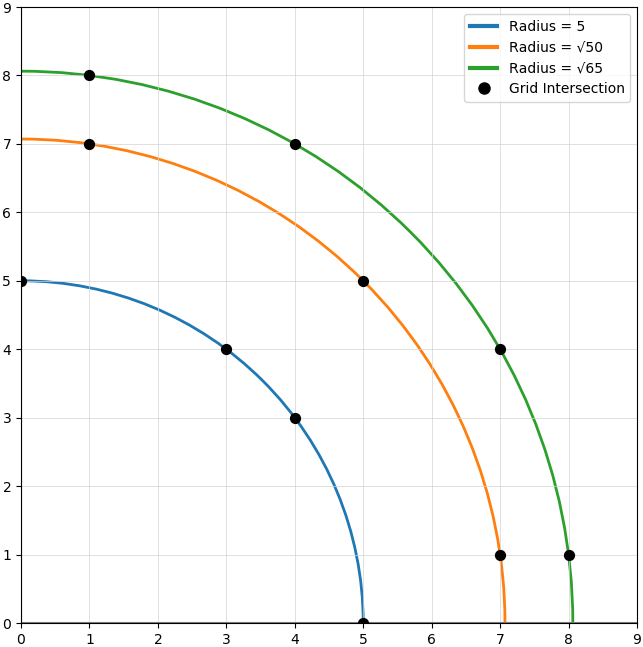}
\caption{Examples of equal sums of two squares.
As explained in the text, $1^2 + 8^2 = 4^2 + 7^2$ gives us $1/7^2 - 1/8^2 = 1/14^2 - 1/56^2$, and
$1^2 + 7^2 = 5^2 + 5^2$ gives us $1/5^2 - 1/7^2 = 1/7^2 - 1/35^2$.
Although $3^2 + 4^2 = 5^2 + 0^2$ is not a case with all positive integers, it is interesting to observe that it gives us $1/20^2 - 1/\infty^2 = 1/12^2 - 1/15^2$,
where we allow the $\infty$ to denote an $E = 0$ ionized state.}
\label{fig2}
\end{figure}

\subsection{All equal sums of two squares arise this way}

\begin{theorem}\label{thm:all-sums}
Every integer solution of
\[
   A^2+B^2=C^2+D^2
\]
arises, up to signs and order, from the four-parameter parametrization
\eqref{eq:ABCD}.
\end{theorem}

\begin{proof}
Define
\[
   z=A+iB,
   \qquad
   w=C+iD.
\]
The equation $A^2+B^2=C^2+D^2$ is the statement
\[
  \Norm(z)=\Norm(w),
\]
where $\Norm(x+iy)=x^2+y^2$ is the Gaussian-integer norm (the number-theoretic norm, which is the square of
the usual complex modulus)~\cite{NivenZuckermanMontgomery1991, ConradGaussianIntegers, IrelandRosen1990}.

The Gaussian integers $\Zi$ form a unique factorization domain (see, for example, Conrad~\cite{ConradGaussianIntegers} or
Niven, Zuckerman, and Montgomery~\cite{NivenZuckermanMontgomery1991}). Let $\delta$
be a greatest common divisor of $z$ and $w$ in $\Zi$. We define $z_0$ and $w_0$ so that
\[
   z=\delta z_0,
   \qquad
   w=\delta w_0.
\]
Since $\Norm(z)=\Norm(w)$, we also have
\[
   \Norm(z_0)=\Norm(w_0).
\]
Because $z_0$ and $w_0$ are coprime in \(\mathbb Z[i]\)
and satisfy $z_0\overline{z_0}=w_0\overline{w_0}$, unique factorization implies that 
\(z_0\) divides \(\overline{w_0}\), so the
prime factors of $w_0$ are, up to units, the conjugates of the prime factors of
$z_0$. Therefore
\[
   w_0=\varepsilon\,\overline{z_0}
\]
for some unit $\varepsilon\in\{\pm1,\pm i\}$.

Absorbing this unit into signs and order, we write
\[
   \delta=r+is,
   \qquad
   z_0=u+iv.
\]
Then
\[
   z=(r+is)(u+iv)
\]
and
\[
   w=(r+is)(u-iv),
\]
again up to multiplication by one of the units $\{\pm1,\pm i\}$. Expanding gives exactly the formulas
\eqref{eq:ABCD}.
\end{proof}

\begin{remark}
This is the analogue, for the equation
\[
   A^2+B^2=C^2+D^2,
\]
of the familiar parametrization of primitive Pythagorean triples. The main
difference is that this four-parameter parametrization is not unique, because
an integer may have several distinct factorizations as a product of two Gaussian
integers [for example, $65=(1+8i)(1-8i)=(4+7i)(4-7i)$,
while a Gaussian-prime factorization is $65=(2+i)(2-i)(3+2i)(3-2i)$]
\end{remark}

\subsection{From equal sums of squares to reciprocal-square identities}

Suppose
\[
   A^2+B^2=C^2+D^2
\]
with $A,B,C,D$ positive. Let $L$ be any common multiple of
$A,B,C,D$. Then
\[
   \left(\frac{A}{L}\right)^2+\left(\frac{B}{L}\right)^2
   =
   \left(\frac{C}{L}\right)^2+\left(\frac{D}{L}\right)^2.
\]
Equivalently,
\begin{equation}\label{eq:reciprocal-from-squares}
   \frac{1}{(L/A)^2}+\frac{1}{(L/B)^2}
   =
   \frac{1}{(L/C)^2}+\frac{1}{(L/D)^2}.
\end{equation}
Thus every case of equal sums of two squares gives an equality of two sums of
reciprocal squares.
Some examples are shown in Figure~\ref{fig2}. The dots on the blue curve have been known for at least 3500 years~\cite{Mansfield2021Si427}.

Conversely, suppose
\[
   \frac{1}{a^2}+\frac{1}{b^2}
   =
   \frac{1}{c^2}+\frac{1}{d^2}.
\]
Let $L$ be a common multiple of $a,b,c,d$, and define
\[
   A=\frac{L}{a},
   \qquad
   B=\frac{L}{b},
   \qquad
   C=\frac{L}{c},
   \qquad
   D=\frac{L}{d}.
\]
Then
\[
   A^2+B^2=C^2+D^2.
\]
Therefore reciprocal-square identities and equal-sum-of-two-squares identities
are the same problem after clearing denominators.

\subsection{All equal-frequency pairs are generated}

\begin{theorem}\label{thm:all-equifrequency}
Every positive-integer equal-frequency transition pair
\[
   N_1\to n_1,
   \qquad
   N_2\to n_2,
\]
which, per definition, satisfies
\[
   \frac{1}{n_1^2}-\frac{1}{N_1^2}
   =
   \frac{1}{n_2^2}-\frac{1}{N_2^2}
\]
is obtained from the Brahmagupta--Fibonacci parametrization
\eqref{eq:ABCD}, after clearing denominators and relabeling terms.
\end{theorem}

\begin{proof}
Starting from an equal-frequency pair, rewrite the equation as
\[
   \frac{1}{n_1^2}+\frac{1}{N_2^2}
   =
   \frac{1}{N_1^2}+\frac{1}{n_2^2}.
\]
Let $L$ be a common multiple of $n_1,N_1,n_2,N_2$, and define
\[
   A=\frac{L}{n_1},
   \qquad
   B=\frac{L}{N_2},
   \qquad
   C=\frac{L}{N_1},
   \qquad
   D=\frac{L}{n_2}.
\]
Then
\[
   A^2+B^2=C^2+D^2.
\]
By Theorem~\ref{thm:all-sums}, the integers $A,B,C,D$ arise from the
Brahmagupta--Fibonacci parametrization, up to signs and order. Hence the
original equal-frequency transition pair arises from the parametrization after
clearing denominators.

Conversely, start with any identity
\[
   A^2+B^2=C^2+D^2
\]
coming from \eqref{eq:ABCD}. If $A \ne C$ and $B \ne D$, then the terms can be labeled so that
\[
   A>C,
   \qquad
   D>B,
\]
then choose any common multiple $L$ of $A,B,C,D$ and define
\[
   n_1=\frac{L}{A},
   \qquad
   N_1=\frac{L}{C},
   \qquad
   n_2=\frac{L}{D},
   \qquad
   N_2=\frac{L}{B}.
\]
The inequalities $A>C$ and $D>B$ imply $N_1>n_1$ and $N_2>n_2$, and we have
\[
   \frac{1}{n_1^2}-\frac{1}{N_1^2}
   =
   \frac{A^2-C^2}{L^2},
\]
and
\[
   \frac{1}{n_2^2}-\frac{1}{N_2^2}
   =
   \frac{D^2-B^2}{L^2}.
\]
Since $A^2+B^2=C^2+D^2$, we have
\[
   A^2-C^2=D^2-B^2.
\]
Therefore the two transition frequencies are equal.
\end{proof}

Of course, any two Pythagorean triples can be used to construct an equal-frequency
example.

\subsection{Example: the transitions 11 to 10 and 55 to 22}

Choose
\[
   r=4,
   \qquad
   s=3,
   \qquad
   u=2,
   \qquad
   v=1.
\]
Then
\[
   ru-sv=5,
   \qquad
   rv+su=10,
\]
and
\[
   ru+sv=11,
   \qquad
   su-rv=2.
\]
Thus
\[
   5^2+10^2=11^2+2^2=125.
\]
Taking $L=110$, we obtain
\[
   \frac{1}{10^2}+\frac{1}{55^2}
   =
   \frac{1}{11^2}+\frac{1}{22^2}.
\]
Moving terms gives the equal-frequency transition identity
\[
   \boxed{
   \frac{1}{10^2}-\frac{1}{11^2}
   =
   \frac{1}{22^2}-\frac{1}{55^2}}
\]
corresponding to
\[
   \boxed{11\to 10,\qquad 55\to 22.}
\]

\subsection{Example: the cascade 35 to 7 and 7 to 5}

A two-step cascade is a special case in which the upper level of one transition
is the lower level of the other. The smallest nontrivial example is
\[
   35\to7,
   \qquad
   7\to5.
\]
This comes from
\[
   7^2+1^2=5^2+5^2=50.
\]
For example, this identity is produced by the Brahmagupta--Fibonacci
parameters
\[
   r=2,
   \qquad
   s=1,
   \qquad
   u=3,
   \qquad
   v=-1.
\]
Taking $L=35$, we get
\[
   \frac{1}{5^2}+\frac{1}{35^2}
   =
   \frac{1}{7^2}+\frac{1}{7^2}.
\]
Equivalently,
\[
   \boxed{
   \frac{1}{5^2}-\frac{1}{7^2}
   =
   \frac{1}{7^2}-\frac{1}{35^2}.}
\]

\subsection{Primitive normalization}

For many purposes one wants a primitive identity, meaning
\[
   \gcd(A,B,C,D)=1.
\]
Starting from the four numbers in \eqref{eq:ABCD}, one may always divide by
their common gcd. If $r+is$ and $u+iv$ are primitive Gaussian integers, meaning
\[
   \gcd(r,s)=1,
   \qquad
   \gcd(u,v)=1,
\]
then the only unavoidable common factor is a possible factor of $2$, occurring
when all four of $r,s,u,v$ are odd. In that case all four expressions in
\eqref{eq:ABCD} are even. Dividing by the common gcd gives the primitive
version.

The equal-frequency transition levels obtained after clearing denominators may
also have a common factor. Dividing all four levels by their common gcd leaves
the equality of transition frequencies unchanged.

\subsection{Summary}

The main points are:
\begin{enumerate}
   \item Equal-frequency transitions are equivalent to equalities of
   two sums of reciprocal squares.

   \item After clearing denominators, equalities of reciprocal-square sums are
   exactly identities
   \[
      A^2+B^2=C^2+D^2.
   \]

   \item Every identity $A^2+B^2=C^2+D^2$ arises from the
   Brahmagupta--Fibonacci parametrization
   \[
      (ru-sv)^2+(rv+su)^2=(ru+sv)^2+(su-rv)^2.
   \]

   \item Therefore the four-parameter Brahmagupta--Fibonacci construction
   generates every positive-integer equal-frequency transition pair, at least
   once, after clearing denominators and relabeling terms.

   \item In Gaussian-integer language, the classification is simply the
   classification of pairs of Gaussian integers with equal norm.
\end{enumerate}

\section{Cascades of Equal-Frequency Transitions}

We study cascades of the following form:
\begin{equation}
    x_0<x_1<x_2,
\end{equation}
with
\begin{equation}
    \frac{1}{x_0^2}-\frac{1}{x_1^2}
    =
    \frac{1}{x_1^2}-\frac{1}{x_2^2}.
\end{equation}
Equivalently,
\begin{equation}
    \frac{1}{x_0^2}+\frac{1}{x_2^2}
    =
    \frac{2}{x_1^2}.
\end{equation}
(Longer cascades of equal-frequency transitions do not occur. This can be proved using a theorem due to Fermat: ``There do not exist four distinct integer squares in arithmetic progression.''
Equivalently, there do not exist four distinct rational squares in arithmetic progression.)
We will see that these length-$2$ cascades naturally correspond to Pythagorean triples.

\begin{figure}[!htb]
\centering
\includegraphics[width=0.7\textwidth]{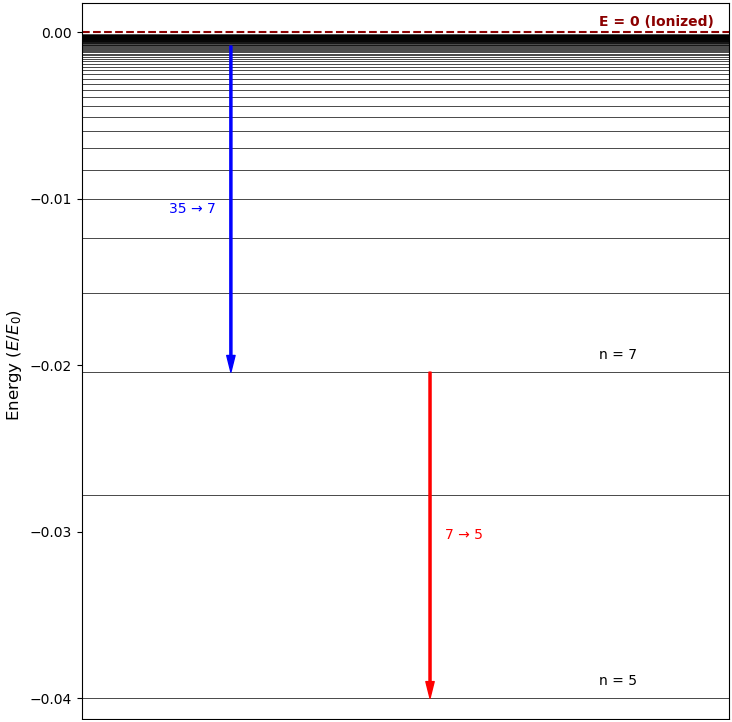}
\caption{An example of a cascade with equal-frequency transitions}
\label{fig3}
\end{figure}

Let $L$ be a common multiple of $x_0,x_1,x_2$, and define
\begin{equation}
    A=\frac{L}{x_0},
    \qquad
    B=\frac{L}{x_1},
    \qquad
    C=\frac{L}{x_2}.
\end{equation}
Since
\[
    x_0<x_1<x_2,
\]
we have
\[
    A>B>C>0.
\]
The cascade condition becomes
\begin{equation}
    A^2-B^2=B^2-C^2,
\end{equation}
or
\begin{equation}
    A^2+C^2=2B^2.
    \label{eq:three-squares}
\end{equation}
Thus length-$2$ cascades are equivalent to three integer squares
\[
    A^2,\;B^2,\;C^2
\]
in arithmetic progression.

\subsection{Parametrization of all length-$2$ cascades}

All primitive integer solutions of
\[
    A^2+C^2=2B^2
\]
come from the standard parametrization of three squares in arithmetic progression.
Choose integers
\[
    m>n>0.
\]
Define
\begin{align}
    A_0 &= m^2+2mn-n^2, \\
    B_0 &= m^2+n^2, \\
    C_0 &= \left|m^2-2mn-n^2\right|.
\end{align}
Then
\begin{equation}
    A_0^2+C_0^2=2B_0^2.
\end{equation}
For non-primitive solutions, an overall multiplier can be incorporated.

For details on this standard parametrization, see
Conrad~\cite{ConradThreeSquaresAP} and Brown, Freedman, and
Shiue~\cite{BrownFreedmanShiue2004}.

\subsection{Relation to Pythagorean triples and rational points on the unit circle}

Equation~\eqref{eq:three-squares} makes it possible to define a rational point on the unit circle. If
\[
    A^2+C^2=2B^2,
\]
define
\begin{equation}
    z=\frac{A+iC}{B(1+i)}.
\end{equation}
Then
\begin{equation}
    |z|^2
    =
    \frac{A^2+C^2}{2B^2}
    =1.
\end{equation}
The calculation
\begin{equation}
    z
    =
    \frac{(A+iC)(1-i)}{2B}
    =
    \frac{A+C}{2B}
    +
    i\frac{C-A}{2B},
\end{equation}
shows that $z$ has rational real and imaginary parts.

Conversely, suppose
\begin{equation}
    z=\frac{p+iq}{r},
    \qquad
    p,q,r\in\Z,
    \qquad
    p^2+q^2=r^2.
\end{equation}
Then the integers
\begin{equation}
    A=p-q,
    \qquad
    C=p+q,
    \qquad
    B=r
\end{equation}
satisfy
\begin{equation}
    A^2+C^2=2B^2.
\end{equation}
After taking absolute values and swapping $A$ and $C$ if necessary, one obtains a
triple of three squares in arithmetic progression.

Thus length-$2$ cascades are equivalent, up to trivial normalizations, to rational
points on the unit circle.
Rational points on the unit circle correspond, after clearing denominators,
to integer solutions of \(a^2+b^2=c^2\); restricting to positive coordinates
gives the usual Pythagorean triples.
[As indicated in the Figure~\ref{fig2} caption, there is also a natural correspondence
between Pythagorean triples
and equal-frequency pairs if we allow the inclusion of an $E = 0$ ionized state.
In summary, a Pythagorean triple naturally gives us two equalities of sums of squares, 
$p^2+q^2=r^2+0^2$ and $(p+q)^2+(p-q)^2=r^2+r^2$.
The former appears in Figure~\ref{fig2} as an example with dots on coordinate axes,
and the latter appears
as an example with a dot on the main diagonal.]

\subsection{Composition law}

The rational points on the unit circle form a group under multiplication of complex numbers.

Suppose two cascades give
\begin{equation}
    z_1=\frac{p_1+iq_1}{r_1},
    \qquad
    z_2=\frac{p_2+iq_2}{r_2},
\end{equation}
with
\begin{equation}
    p_1^2+q_1^2=r_1^2,
    \qquad
    p_2^2+q_2^2=r_2^2.
\end{equation}
Then
\begin{equation}
    z_1z_2
    =
    \frac{(p_1p_2-q_1q_2)+i(p_1q_2+p_2q_1)}{r_1r_2}.
\end{equation}
The numerator again satisfies a Pythagorean identity:
\begin{equation}
    (p_1p_2-q_1q_2)^2+(p_1q_2+p_2q_1)^2
    =
    (r_1r_2)^2.
\end{equation}
This is precisely the Brahmagupta--Fibonacci identity.

As an example, let us compose the smallest cascade with itself.
For the smallest cascade, we may start with
\[
    z=\frac{4-3i}{5}.
\]
Squaring gives
\begin{equation}
    z^2
    =
    \left(\frac{4-3i}{5}\right)^2
    =
    \frac{7-24i}{25}.
\end{equation}
Thus
\[
    p=7,
    \qquad
    q=-24,
    \qquad
    r=25.
\]
Using
\begin{equation}
    A=p-q,
    \qquad
    C=p+q,
    \qquad
    B=r,
\end{equation}
we get
\begin{equation}
    A=31,
    \qquad
    C=-17,
    \qquad
    B=25,
\end{equation}
\begin{equation}
    31^2+17^2=2\cdot25^2.
\end{equation}
Define
\begin{equation}
    L=\lcm(31,25,17)=13175.
\end{equation}
Then
\begin{align}
    x_0&=\frac{13175}{31}=425,\\
    x_1&=\frac{13175}{25}=527,\\
    x_2&=\frac{13175}{17}=775.
\end{align}
Therefore
\begin{equation}
\boxed{
    527\to425
    \qquad\text{and}\qquad
    775\to527
}
\end{equation}
have equal frequency.

Equivalently,
\begin{equation}
    \frac{1}{425^2}-\frac{1}{527^2}
    =
    \frac{1}{527^2}-\frac{1}{775^2}.
\end{equation}

\section{Diophantine equations related to the nonrelativistic Schr\"odinger Equation for hydrogen}
\label{sec_SE}

In this section, we derive some results concerning descriptions of hydrogen that go beyond Bohr's 1913 model.
We will discuss the nonrelativistic Schr\"odinger equation for hydrogen.
For a standard undergraduate treatment of this, including separation in spherical coordinates and the
quantum numbers $n,l,m,$ see Griffiths and Schroeter~\cite{GriffithsSchroeter2018}.
We will also discuss the Langer modification~\cite{Langer1937, Morehead1995, Morehead2005,
MorikawaOgawa2025}.

In classical mechanics, the eccentricity of a Kepler orbit may be
written in terms of the energy \(E\) and angular momentum \(L\) as discussed below.
Thus the eccentricity defines a function on phase space. In the quantum
problem we may use the same expression to define an operator,
\[
\widehat{\epsilon^2}
=
1+\frac{2\hat H \hat L^2}{\mu\kappa^2}.
\]
This definition should not be interpreted as saying that the electron moves
along a definite elliptical path. It is simply the definition of an operator,
motivated by the existence of a classical function. There is no
operator-ordering ambiguity, because \(\hat H\) and \(\hat L^2\) commute.

There is also a closely related alternative based on the Laplace--Runge--Lenz
vector. Classically, the square of this vector is proportional to
\(\epsilon^2\). In the quantum problem, however, the Laplace--Runge--Lenz
operator involves products of noncommuting operators, and its Hermitian
definition requires a choice of ordering\cite{Weinberg2015,SakuraiNapolitano2020,
BanderItzykson1966I,SkinnerPQM,Binney2013physics}. As a result, the square of the
quantum Laplace--Runge--Lenz vector differs from the direct
\(\hat H,\hat L^2\) construction by an additive term of order \(\hbar^2\).
In order to look at this in more detail, let us begin with
the standard Hermitian quantum Runge--Lenz vector,
\[
\hat{\mathbf A}
=
\frac{1}{2\mu}
\left(
\hat{\mathbf p}\times \hat{\mathbf L}
-
\hat{\mathbf L}\times \hat{\mathbf p}
\right)
-
\kappa\,\frac{\hat{\mathbf r}}{r}.
\]
The symmetrization is needed because \(\hat{\mathbf p}\) and
\(\hat{\mathbf L}\) do not commute.
Based on this expression, the following operator identity can be derived:
\[
\hat A^2
=
\kappa^2
+
\frac{2\hat H}{\mu}
\left(
\hat L^2+\hbar^2
\right).
\]
Thus we may define an operator
\[
\widehat{\epsilon_A^2}
=
\frac{\hat A^2}{\kappa^2}
=
1
+
\frac{2\hat H}{\mu \kappa^2}
\left(
\hat L^2+\hbar^2
\right).
\]
Acting on \(|n\ell m\rangle\), this gives
\[
\widehat{\epsilon_A^2}|n\ell m\rangle
=
\left(
1-\frac{\ell(\ell+1)+1}{n^2}
\right)
|n\ell m\rangle ,
\]
and the eigenvalue appears in the table below (see Method~3), 
which gives an overview of what we will do in this section:
\[
\begin{array}{l|c|c|c}
\text{Method}
&
\epsilon^2
&
\epsilon \text{ rational?}
&
a/b \text{ rational?}
\\
\hline
\text{Method 1: direct } \hat H,\hat L^2
&
1-\dfrac{\ell(\ell+1)}{n^2}
&
\text{sometimes rational}
&
\text{only for } \ell = 0
\\[10pt]
\text{Method 2: Langer}
&
1-\dfrac{(\ell+1/2)^2}{n^2}
&
\text{not rational}
&
\text{always rational}
\\[10pt]
\text{Method 3: Runge--Lenz}
&
1-\dfrac{\ell(\ell+1)+1}{n^2}
&
\text{sometimes rational}
&
\text{only for } \ell = 0
\end{array}
\]

\subsection{Eccentricity and axis ratios}

For a classical Kepler orbit, the eccentricity
\(\epsilon\) may be written in terms of the energy \(E\) and angular momentum
\(L\) as
\[
\epsilon^2
=
1+\frac{2EL^2}{\mu\kappa^2},
\]
where \(\mu\) is the reduced mass, and we use a value for \(\kappa\) appropriate for
a hydrogenic problem,
\[
\kappa=\frac{Ze^2}{4\pi\epsilon_0}.
\]
For bound hydrogenic motion,
\[
E_n=-\frac{\mu\kappa^2}{2\hbar^2 n^2}.
\]
Substitution gives
\[
\epsilon^2
=
1-\frac{L^2}{\hbar^2 n^2}.
\]
Thus the numerical properties of the associated eccentricity depend on the
choice of quantum angular-momentum prescription.

First, using the result from the Schr\"odinger Equation
\[
L^2=\ell(\ell+1)\hbar^2,
\]
one obtains
\[
\epsilon^2
=
1-\frac{\ell(\ell+1)}{n^2},
\]
which is shown in the row for Method~1 in the table above.
Hence \(\epsilon\) is rational precisely when
\[
n^2-\ell(\ell+1)=k^2
\]
for some integer \(k\). For example,
\[
n=4,\qquad \ell=3
\]
gives
\[
\epsilon^2
=
1-\frac{3\cdot 4}{4^2}
=
\frac14,
\]
and we get \(\epsilon=1/2\).

Next, using the Langer modification
(Method~2 in the table above)
\[
L=\left(\ell+\frac12\right)\hbar,
\]
we get
\[
\epsilon^2
=
1-\frac{(\ell+1/2)^2}{n^2}
=
\frac{(2n)^2-(2\ell+1)^2}{(2n)^2}.
\]
Thus rationality of \(\epsilon\) would require an integer solution of
\[
(2\ell+1)^2+k^2=(2n)^2.
\]
By looking at this equation modulo \(4\), we can see that no solutions are possible:
the first term is congruent to
\(1\) modulo \(4\), while \(k^2\) is congruent to either \(0\) or \(1\) modulo
\(4\), and the right-hand side is congruent to \(0\) modulo \(4\). Therefore
the Langer modification gives no rational eccentricities.

Next we look at the ratio of the major to minor axes. Let \(a\) and \(b\) denote
the semi-major and semi-minor axes of the classical ellipse. Since
\[
\frac{b}{a}=\sqrt{1-\epsilon^2},
\]
we have
\[
\frac{b}{a}
=
\frac{L}{\hbar n}.
\]
With the Langer modification this gives
\[
\frac{b}{a}
=
\frac{\ell+1/2}{n},
\qquad
\frac{a}{b}
=
\frac{2n}{2\ell+1},
\]
so the major-to-minor axis ratio is always rational.

For Method~1, one instead obtains
\[
\frac{b}{a}
=
\frac{\sqrt{\ell(\ell+1)}}{n}.
\]
For \(\ell\geq 1\), the product \(\ell(\ell+1)\) is not a perfect square,
because it lies strictly between two consecutive perfect squares,
\[
\ell^2 < \ell(\ell+1) < (\ell+1)^2.
\]
Thus Method~1 gives no nondegenerate rational major-to-minor axis ratios.

We reach the same conclusion for Method~3 because
for \(\ell\geq 1\), the expression \(\ell(\ell+1)+1\) is not a perfect square;
it lies strictly between two consecutive perfect squares,
\[
\ell^2 < \ell(\ell+1)+1 < (\ell+1)^2.
\]

Method~3 leads to the factorization problem
\[
(n-k)(n+k)=\ell^2+\ell+1.
\]
First, let us regard the value of \(\ell\) as given and ask for the possible values of \(n\).
Since \(\ell^2+\ell+1\) is odd, any factorization of it into two positive
factors gives integer values for \(n\) and \(k\). One possibility that is always
available (even if \(\ell^2+\ell+1\) is prime) is
\[
n-k=1,
\qquad
n+k=\ell^2+\ell+1,
\]
which results in
\[
n=\frac{\ell^2+\ell+2}{2}.
\]
The calculation 
\[
\frac{\ell^2+\ell+2}{2}-\ell
=
\frac{\ell^2-\ell+2}{2}
=
\frac{\ell(\ell-1)+2}{2}>0.
\]
shows that this \(n\) value is greater
than \(\ell\), so it is a valid \(n\) value.

Next, let us regard the value of \(n\) as given and ask for the possible values of \(\ell\).
We begin with the equation
\[
n^2=k^2 + \ell^2+\ell+1,
\]
and rewrite it as
\[
4n^2-3=(2k)^2 + (2\ell+1)^2.
\]
We see that this is
equivalent to asking whether \(4n^2-3\) is representable as a sum of two
squares. We do not pursue this question here; it belongs naturally to the
number-theoretic study of sums of two squares and values of quadratic
polynomials.

Thus the Runge--Lenz version provides a third natural 
integer problem coming from the eccentricity problem.

\begin{remark}
We do not suggest that rational values of the eccentricity have direct
spectroscopic significance. Rather, 
the study of Diophantine equations is fundamentally related to the
mathematical framework that emerges from quantum mechanics,
and it is of interest to study
the different ways in which the classical Kepler ellipse may be
associated with the quantum numbers \(n\) and \(\ell\). The direct
\(\hat H,\hat L^2\) construction, the Langer modification, and the
Runge--Lenz construction differ only by small changes to integer expressions,
but these changes
lead to rather different Diophantine problems.
\end{remark}

\subsection{Connection with equal sums of two squares}

The condition for rational eccentricity in Method~1 was
\[
n^2-\ell(\ell+1)=k^2.
\]
This may be rewritten as
\[
4n^2-4k^2=4\ell(\ell+1)=(2\ell+1)^2-1.
\]
Thus
\[
(2\ell+1)^2+(2k)^2=(2n)^2+1^2.
\]
This is an equality of two sums of two squares (for a thorough analysis
of this equation see the literature\cite{Frink1987AlmostPythagoreanTriples,
AntalanTomenes2015AlmostPythagoreanTriples} on ``almost Pythagorean triples,'').
Thus the rational-eccentricity problem is a special ``slice'' of the
equal-sums-of-two-squares problem considered earlier in this paper.

A useful construction comes from the Brahmagupta--Fibonacci identity
\[
(ac+bd)^2+(ad-bc)^2=(ac-bd)^2+(ad+bc)^2.
\]
If we impose the additional condition
\[
ac-bd=1,
\]
then the identity becomes
\[
(ac+bd)^2+(ad-bc)^2=1^2+(ad+bc)^2.
\]
Comparing this with
\[
(2\ell+1)^2+(2k)^2=1^2+(2n)^2,
\]
we may set
\[
2\ell+1=ac+bd,\qquad 2k=ad-bc,\qquad 2n=ad+bc.
\]
Since \(ac-bd=1\), we have
\[
ac=bd+1,
\]
and therefore
\[
ac+bd=2bd+1.
\]
Thus
\[
\ell=bd,
\]
and the corresponding values of \(n\) and \(k\) are
\[
n=\frac{ad+bc}{2},\qquad k=\frac{ad-bc}{2}.
\]
In summary, we have
\[
\boxed{
\ell=bd,\qquad
n=\frac{ad+bc}{2},\qquad
k=\frac{ad-bc}{2},
\qquad
ac-bd=1.
}
\]
The parity condition that \(ad\) and \(bc\) have the same parity ensures that
\(n\) and \(k\) are integers.

It is possible to construct simple polynomial families of solutions. For example, choose
\[
a=4uv+1,\qquad b=2u,\qquad c=1,\qquad d=2v.
\]
Then
\[
ac-bd=(4uv+1)-4uv=1,
\]
and the construction gives
\[
\boxed{
\ell=4uv,\qquad
n=v(4uv+1)+u,\qquad
k=v(4uv+1)-u
}
\]

There is also a companion family. Choose
\[
a=2v,\qquad b=1,\qquad c=2u,\qquad d=4uv-1.
\]
Then
\[
ac-bd=4uv-(4uv-1)=1,
\]
and the construction gives
\[
\boxed{
\ell=4uv-1,\qquad
n=v(4uv-1)+u,\qquad
k=v(4uv-1)-u
}
\]
is another two-parameter family. The small example
\[
n=4,\qquad \ell=3,\qquad k=2
\]
is obtained by taking \(u=v=1\).

\begin{acknowledgments}
The author acknowledges the use of ChatGPT, developed by OpenAI, for
exploratory mathematical discussion, literature-search suggestions, and
assistance in organizing the exposition. ChatGPT was not used as an
authoritative source; all mathematical arguments, calculations, references,
and conclusions were independently checked by the author.
\end{acknowledgments}

\newpage

\vspace{6pt} 

\bibliography{main}
\bibliographystyle{apsrev4-2}

\end{document}